\documentclass[12pt]{article}
\usepackage{verbatim}
\usepackage{amsfonts}
\usepackage{graphics}
\usepackage{amsmath}
\usepackage{times}
\usepackage{appendix}
\usepackage{color}
\usepackage{soul}
\usepackage{mathtools}
\usepackage{enumerate}
\usepackage{fancyhdr,latexsym,amsmath,amsfonts,amssymb,amsbsy,amsthm,url}
\usepackage[margin=0.5in,footskip=0.25in]{geometry}
\usepackage{graphics,graphicx,epsfig}
\usepackage{breqn}
\usepackage{caption,subcaption,float,subfloat}
\usepackage{pdflscape}
\usepackage{hyperref}
\usepackage[ruled,vlined]{algorithm2e}
\newtheorem{theorem}{Theorem}[]
\newtheorem{lemma}{Lemma}[]

\newtheorem{problem}{Problem}[]
\usepackage[english]{babel}
\usepackage[dvipsnames]{xcolor}
\makeatletter

\renewcommand{\section}{
	\@startsection
	{section}% name
	{1}% level
	{0pt}% indent
	{1.1\baselineskip}% beforeskip
	{0.2\baselineskip}% afterskip
	{\sc \centering}% style
}

\renewcommand{\subsection}{
	\@startsection
	{subsection}% name
	{1}% level
	{0pt}% indent
	{1.1\baselineskip}% beforeskip
	{0.2\baselineskip}% afterskip
	{\sc \centering}% style
}

\renewcommand{\subsubsection}{
	\@startsection
	{subsubsection}% name
	{1}% level
	{0pt}% indent
	{1.1\baselineskip}% beforeskip
	{0.2\baselineskip}% afterskip
	{\sc \centering}% style
}

\makeatother

\usepackage[flushleft]{threeparttable}
\usepackage{rotating,booktabs,multirow}
\usepackage{colortbl}
\usepackage{makecell,cellspace,caption}

\begin{document}
	
\title{\large\sc Does limited liability reduce leveraged risk?: The case of loan portfolio management}
\normalsize
\author{
\sc{Deb Narayan Barik} \thanks{Department of Mathematics, Indian Institute of Technology Guwahati, Guwahati-781039, India, e-mail: d.narayan@iitg.ac.in}
\and 
\sc{Siddhartha P. Chakrabarty} \thanks{Department of Mathematics, Indian Institute of Technology Guwahati, Guwahati-781039, India, e-mail: pratim@iitg.ac.in, Phone: +91-361-2582606}}

\date{}
\maketitle
\begin{abstract}
	
Return-risk models are the two pillars of modern portfolio theory, which are widely used to make decisions in choosing the loan portfolio of a bank. Banks and other financial institutions are subjected to limited liability protection. However, in most of the model formulation, limited liability is not taken into consideration. Accordingly, to address this, we have, in this article, analyzed the effect of including it in the model formulation. We formulate four models, two of them are maximizing the expected return with risk constraint, including and excluding limited-liability, and other two are minimization of risk with threshold level of return with and without limited-liability. Our theoretical results show that the solutions of the models with limited-liability produce better results than the others, in both minimizing risk and maximizing expected return. It has less risky investment than the other portfolio that solves the other model. Finally, an illustrative example is presented to support the theoretical results obtained.


\end{abstract}

\section{Introduction}

Bank's leverage to a great extent may end up playing a vital role in triggering a financial crisis, as was the case with the spectacular financial collapse of 2008, which is now widely attributed to excessive leveraging by the prominent investment banks. Given the fact that many of these financial institutions which were over-leveraged, had maintained a healthy level of capital requirement compliance, the regulators, namely the Basel Committee on Banking Supervision (BCBS) were prompted by this, to set a target upper bound on the extent to which a bank can be leveraged \cite{Basel2014}. Accordingly, they introduced a non-risk based capital measure, namely, the Leverage Ratio, to insulate the banks from over-leveraging, by way of capital requirements. Consequently, the Leverage Ratio is defined as, 
\[\text{Leverage Ratio}:=\frac{\text{Capital Measure}}{\text{Exposure Measure}}.\]
The numerator of ``Capital Measure'' in the definition is the Tier 1 capital of the risk based capital framework, whereas the denominator of ``Exposure Measure'' is the sum of the exposures from the on-balance items, off-balance sheet items, derivatives and securities financing transactions (SFT).

The well established approaches for determination of capital requirements, not withstanding, the possibility of higher returns drives banks to greater risk exposure. Further, it is possible that banks will tend to underplay the extent of the risk exposure, to the supervisors, which in the worst case, may even lead to bankruptcy. This problem and its step-by-step solution was discussed in \cite{Blum2008}. The problem of limited ability of supervisors to decipher whether the bank disclosures are honest or not, is shown to have an enforceable solution, by way of imposition of a risk-independent Leverage Ratio restriction. A World Bank report \cite{D2009} has elaborated upon the concept of leverage, and also the necessity of the Leverage Ratio, to complement the already existent capital requirement framework, with the caveat of it (Leverage Ratio) being one of the several policy tools, in the paradigm of assessing the leverage buildup in a financial institution. In this context, Hildebrand \cite{Hildebrand2008} highlighted the benefits of a Leverage Ratio, while recognizing the shortcomings. The benefits include the complementary non-risk based nature and the simplicity, in terms of definition, application and monitoring, while the shortcomings include off-balance sheet exposure, profitability and pro-cyclicality. Dell'Ariccia et al. \cite{Dellariccia2014} in their work related the low interest rate scenario with the leverage, as well as risk-taking by the banks. The authors concluded that when banks are accorded the leeway of determining their capital structure, reduction in interest rates lead to increased leverage, and consequently, higher risk, provided the loan demand function has a diminishing slope. However, in case of fixed capital structure, the impact is contingent on the extent of the leverage. Subsequently, Smith et al. \cite{Acosta2020} worked on the bank's decision problem, contingent on the Leverage Ratio restriction. In the model, the capital holding is considered, with all investment done in one risky asset, and all payments being done from the bank's side. It was shown that every bank has a tendency to hold less capital, which is otherwise very evident. Further, one constant $\widehat{k}$ was derived for each bank, depending upon the assets of the bank, thereby establishing a relation with Leverage Ratio and risk taking. Finally, it was noted that relative to a solely risk-based capital framework, the imposition of the Leverage Ratio requirement leads to lower probabilities of bank failure.

In \cite{Kiema2014} and \cite{Repullo2004}, the authors have categorized the loan types as high risk loans and low risk loans. Unlike the previous works, in this case, instead of considering the analysis of the portfolio of one bank, they have established the stability of the banks, by the expected number of bank failures, among the banks operating in the market. Further in the equilibrium (where the demand of loans and the supply of loans are equal, it is defined as the zero net value of the bank), they determined the relation between the Leverage Ratio and the banks' portfolio, with low risk and high risk loans. Finally, the papers also analyzed that higher Leverage Ratio increases the stability of the bank. As the Leverage Ratio increases, the number of high risk loans decreases in the portfolio of the bank. Sale of bank loans can often be motivated by limitations resulting from regulatory requirements \cite{Carlstrom1995}. In case of unregulated banks, there is enough incentive to extend loans and then sell off these loans to other banks, rather than adopt a more traditional approach of accepting deposits to fund these loans. Banks can extend credit only to the extent that regulators allow for. In this scenario, it is the approach of loan sales that is attractive for such banks (thereby going beyond the permissible regulatory limits) and also for the banks which have room (from capital requirement perspective) for purchasing these loans. In practice, this amounts to capital buffer of non-local banks to support local projects routed through local banks.

The aspects of leverage and the consequent risk is intricately linked to risk-return paradigm of the portfolio of loans held by the bank \cite{Mencia2012}.
Accordingly, the distribution of loan portfolios was studied sector-wise to encapsulate the cyclical characteristics of different types of loans and yields. The classical Markowitz approach is applied with a Value-at-Risk (VaR) constraint (to accommodate regulatory requirements), as well as the relationship maximization of the utility function. Further, the model was also used to capture information about default correlations. In addition, this exhaustive study also included the determination of a no-arbitrage principle driven pricing of loans, taking into account the credit risk associated with the debtor. An optimization approach to the construction of a loan portfolio using VaR and Conditional VaR (CVaR) constraints is presented in \cite{Ming2015}. To this end, the basic approach of Lagrangian algorithm is employed, in order to determine the efficient frontier. An alternative approach of minimization of credit risk, in terms of expected loss is empirically analyzed in \cite{Cho2012}. Accordingly, a methodology based on the Large Deviation Theory (LDT) is used for portfolio optimization, by taking into account, the heterogeneity of risk characteristic across different geographical locations. The key takeaway was the demonstration of a significant improvement in the performance of this approach, vis-a-vis other benchmark portfolios, with this improvement being achieved in terms of enhanced excess return and reduced credit risk. For a detailed description of portfolio management, in presence of default risk, one may refer to the KMV document \cite{Kealhofer1998}. The authors of the work, enumerated the key aspects of this problem, and covered models of default risk, default correlation and value correlation. This is followed by risk contribution and its relation to optimal diversification, as well as economic capital. They concluded with an accurate and detailed description of the risk of losses experienced in loan portfolios, by considering different ``sub-portfolios'', of a typically very large portfolio of banks. 

From a historical perspective, limited liability, in practice is achieved by way of a private contractual setup \cite{Carney1998}. It offers the advantage of safeguard of passive investors, from creditors, in the event of bankruptcy. In today's economy, the implications of limited liability, as a source of moral hazard, is significantly evident \cite{Djelic2013}. The authors of the work \cite{Djelic2013} strongly suggest the structural connect between moral hazard and limited liability, and highlighted the disruptive socio-economic consequences of the same. The question of dynamic moral hazard, resulting from the un-observable effort of an agent with limited liability, in terms of low frequency high magnitude losses is examined in \cite{Biais2010}. The setup consisted of the agent and the principal, who (unlike the agent) has unlimited liability. An optimal approach, in terms of payments to the agent, is contingent on the good performance of the agent, in absence of which the payments are stopped. This is turn leaves room for both the extreme possibilities of the firm size diminishing to zero or experiencing unbounded growth. Limited liability, while acting as an incentive for investors, and facilitating economic growth, can have adverse consequences, such as risk taking tendencies, thereby causing economic loss \cite{Simkovic2018}. Accordingly, the consequences of the latter can be sought to be mitigated through regulatory mechanisms (including capital requirements) and statutory insurance. An analysis of limited liability for insurance markets is studied in \cite{Boonen2019}, by considering limited liability protection for non-life insurers. In particular, the case considered in the study is one where there is an exchangeable nature of insurance risk, in case of policy holders. The author then goes on to establish the existence of a partial equilibrium in the insurance sector.

\section{The Models}

The risk management of a loan portfolio in a bank is an increasingly complex exercise, due to several factors. Banks have been extended the legal rights emanating from the concept of limited liability, which effectively means that the losses cannot exceed the net value of the bank. In the event of the net worth being less than the debt, the owners do not  recover anything, as a result of the bank having gone bankrupt. However, if the net value at the end is more than debt, then the owners are entitled to the remaining assets after the payout have been made to the creditors of the bank. While most of the literature on managing a loan portfolio does the modeling in the regime of the standard risk-return model, in some of the works, the notion of limited liability has been considered for the model setup. While the models due to \cite{Blum2008,Dellariccia2014,Acosta2020,Kiema2014,Repullo2004} mentioned about the limited liability considerations, however they do not extend this discussion to a narrative on the advantages and disadvantages resulting from the usage of limited liability in the model paradigm. A summary of the relevant work in the risk-return framework for the loan portfolio decision problem is presented in the Table \ref{One_Table_One}, which also highlights the novelty and advancement of this work vis-a-vis the existing literature. 

In this article, we show that including limited liability plays a vital role in optimizing the expected return, as well as in reducing the expected and the unexpected loss of the portfolio. Accordingly, we formulate two sets of problems, by the consideration of expected return and risk, which are the two pillars of modern portfolio theory.

\begin{table}[h]
\centering 
\begin{tabular}{ccccc}
\hline
Max of exp return with- & Min of risk with- & Max of exp return with- & Min of risk & Source \\
out lim liability&  out lim liability & lim liability & with lim liability & \\
\hline
No & No  & Yes & No &\cite{Blum2008}, \cite{Dellariccia2014}, \cite{Acosta2020}, \cite{Kiema2014}, \cite{Repullo2004}.\\
\hline
No & Yes & No & No & \cite{Mencia2012},\cite{Ming2015},\cite{Cho2012}.\\
\hline
Yes & No & No & No & \cite{Carlstrom1995}, \cite{Oladejo2020}.\\
\hline
Yes & Yes  & Yes & Yes & This work \\
\hline
\end{tabular}
\caption{Comparative literature highlighting the novelty of the work}
\label{One_Table_One}
\end{table}

\subsection{Models with Risk-Return}

Before going to the modelling setup (for a loan portfolio of size $n$), we introduce the parameters that are going to be used in the risk-return framework,
as given in Table \ref{One_Table_Two}.
\begin{table}[h]
\centering 
\begin{tabular}{cc}
\hline
Symbol & Meaning  \\
\hline
$k_{lev}$& Leverage Ratio.\\
\hline
$K(x)$ & Internal Ratings Based (IRB) Capital Requirement for the portfolio $x$. \\
\hline
$\delta$ & Opportunity cost of the capital (equity). \\
\hline
$X_{x}$ & Realizations for the loan portfolio $x$. \\
\hline
$\rho(x)$ & Risk measure for the portfolio $x$ (EL, UL etc.)\\
\hline
$\mu$ & Lower-bound on Expected Return \\
\hline
$\theta$ & Upper-bound on the risk. \\
\hline
\end{tabular}
\caption{List of symbols used in the model}
\label{One_Table_Two}
\end{table}

\begin{problem}{Maximization of Expected Return without Limited Liability:}\label{One_Prob_One}\\
\[\max_{x,k}\left[E\left[X_{x}\right]-\left(1-k\right)-\delta k\right],\]
subject to,
\[k \geq \max \left(k_{lev},K(x)\right),~\sum\limits_{i=1}^{n} x_{i}=1,~x_{i} \geq 0~\forall~ i=1:n,~\rho(x) \leq \theta.\]
Here $\displaystyle{E(X)=\sum\limits_{i=1}^{n}x_{i}R_{i}}$, where $x_{i}$ and $R_{i}$ are the weights and the expected return of the $i$-th loan in the portfolio. The objective functional (motivated by \cite{Kiema2014}) to be maximized involves the maximization of the expected return minus a function of $k$ (which is the larger of the Leverage Ration and the IRB based capital requirement). In addition, the usual conditions of sum of weights being equal to one, no-short selling being permissible, and an upper bound on the risk, are applicable. Finally, we denoted by $\left(x_{1},k_{1}\right) \in \mathbb{F}_{1}$, the solution of the optimization problem, with $\mathbb{F}_{1}$ being the feasible region for this optimization problem.
\end{problem}
\begin{problem} 
{Minimization of Risk without Limited Liability:}\label{One_Prob_Two}\\
\[\min_{x,k}\rho(x),\]
subject to, 
\[k \geq \max \left\{k_{lev},K(x)\right\},~\sum\limits_{i=1}^{n} x_{i}=1,~ x_{i}\geq 0~\forall~ i=1:n,~E\left[X_{x}\right]-\left(1-k\right)-\delta k \geq \mu.\]
Here, the constraints on $k$, the sum of weights being equal to one and no-short selling being permissible, are akin to that of Problem \ref{One_Prob_One}. Also the expected excess return over $\left((1-k)+\delta k\right)$ is required to be at least a threshold value, denoted by $\mu$. Finally, we denote by $\left(x_{2},k_{2}\right) \in \mathbb{F}_{2}$ be the solution the problem, with $\mathbb{F}_{2}$ being the feasible region for this optimization problem.
\end{problem}

\subsection{Model With Limited Liability}
Before going to the modelling aspects of limited liability, let us discuss some important assumptions.
\begin{enumerate}[(A)]
\item It is observable \cite{LBBW} that banks typically will not invest in loans with the default probability $p \geq 0.2$
(which can also be concluded by analyzing the forms of Unexpected Loss). As reported in \cite{Landesbank2008}, the loans in investment grade and speculative grade have the maximum probability of default being $0.2$. Also the article
\cite{Featherstone2006} notes that for the S\&P Rating with highest default probability, namley the rating of ``D-'', the default probability lies between $0.18$ and $0.20$. 
\item If the bank invests all its money in a single loan, then if the expected return from risky loan, after paying the liabilities and the opportunity cost of the equity, is higher than that of the safe loan, then the bank chooses the risky loan in their portfolio, which mathematically translates to,
\[R_{H}-\left(1-k_{H}^{'}\right)-\delta k_{H}^{'} > R_{L} -(1-k_{L}^{'})-\delta k_{L}^{'}.\]
Here $R_{H}$ and $R_{L}$ are the expected return from high risk and low risk loans, respectively, where $k_{H}^{'}=\max\left\{k_{H},k_{lev}\right\}$ and $ k_{L}^{'}=\max\left\{k_L,k_{lev}\right\}$, are the capital holdings for high risk and low risk loans, respectively. Further, $k_{H}$ and $k_{L}$ are the capital requirements, based on the Internal Ratings Bases (IRB) approach.
\item The loans are uncorrelated among themselves.
\end{enumerate}
Accordingly, we are now in a position to define the following two problems, involving limited liability.
For this purpose we slightly modify the objective functional, motivated by \cite{Acosta2020}, for Problems \ref{One_Prob_Three} and \ref{One_Prob_Four}.
\begin{problem}{Maximization of Profit with Limited Liability:}\label{One_Prob_Three}\\
\[\max_{x,k}E\left[\max\left(\left(X_{x}-(1-k)\right),0\right)\right]-\delta k,\]
subject to
\[k \geq \max \left(k_{lev},K(x)\right),~
\sum\limits_{i=1}^{n} x_{i}=1,~ x_{i} \geq 0~\forall~ i=1:n,~\rho(x) \leq \theta.\]
Let $\left(x_{1L},k_{1L}\right) \in \mathbb{F}_{1L}$ be the solution of the problem, with $\mathbb{F}_{1L}$ being the feasible region for this optimization problem.
\end{problem}

\begin{problem}{Minimization of Risk with Limited Liability:}\label{One_Prob_Four}\\
\[\min_{x,k}\rho(x),\]
subject to 
\[k \geq \max \left\{k_{lev},K(x)\right\},~\sum\limits_{i=1}^{n} x_{i}=1,~ x_{i}\geq 0~\forall~ i=1:n,~E\left[\max\left\{(X_{x}-(1-k)),0\right\}\right]-\delta k \geq \mu.\]
Here $\mu$ is the threshold value of the expected. Let $\left(x_{2L},k_{2L}\right) \in \mathbb{F}_{2L}$ be the solution of the problem, with $\mathbb{F}_{2L}$ being the feasible region for this optimization problem.
\end{problem}

\section{Results}

Here we begin with an important theorem.
\begin{theorem}
\label{One_Theo_One}	
Investing more in risky asset increases Expected Loss and Unexpected Loss.
\end{theorem}
\begin{proof}
In this proof we use the formula used by  \cite{Kealhofer1998}, for the expected loss, 
\[EL_{P}=\sum\limits_{i}x_{i}p_{i}\lambda_{i},\] 
and for the unexpected loss,
\[UL_{P}=\sqrt{\sum\limits_{i}\sum\limits_{j} x_{i}x_{j} \rho_{ij} UL_{i}UL_{j}},\]
where $\displaystyle{UL_{i}=\lambda_{i}\sqrt{p_{i}(1-p_{i})}}$ and $\lambda_{i}$ and $p_{i}$ are the Loss Given Default and Probability of Default, respectively, for the $i$-th loan. So it is clear that risky loans have large $\lambda$, as well as $p_{i}$, and hence increasing either or both of them, for a risky loan, also increases both EL and UL.
\end{proof}	

Now we show that including limited liability in the model leads to out-performance of the models without limited liability. First we show the minimization of risk model with and without limited liability.
\begin{theorem}
\label{One_Theo_Two}	
In case of minimizing risk, Problem \ref{One_Prob_Four} outperforms Problem \ref{One_Prob_Two}.
\end{theorem}
\begin{proof}
For a particular portfolio $x$ and a particular $k$, we have,
\[E\left[\max\left\{(X_{x}-(1-k)),0\right\}\right]-\delta k \geq  E\left[X_{x}\right]-(1-k)-\delta k.\]
Therefore $\mathbb{F}_{2L}$ contains more points than $\mathbb{F}_{2}$. In other words, $\mathbb{F}_{2} \subseteq \mathbb{F}_{2L}$. Hence,
\[\min~\left\{\rho(x)|(x,k) \in \mathbb{F}_{2L}\right\} \leq \min \left\{\rho(x)|(x,k) \in \mathbb{F}_{2}\right\}.\]
\end{proof}
Now, in order to prove that Problem \ref{One_Prob_Three} outperforms Problem \ref{One_Prob_One}, we first establish the following Lemmas.
\begin{lemma}
\label{One_Lemma_One}
In the solution of Problem \ref{One_Prob_One}, $k_{1}=k_{1}^{'}$ where $k_{1}^{'}=\max\left(K(x_1),k_{lev}\right)$. 
\end{lemma} 
\begin{proof}
$E\left[X_{x}\right]-(1-k)-\delta k$ is a monotonically decreasing function of $k$ (keeping $x$ fixed). So if $k_{1} > k_{1}^{'}$, then we get, 
\[E\left[X_{x_{1}}\right]-(1-k_{1}^{'})-\delta k_{1}^{'} > E\left[X_{x_1}\right]-(1-k_{1})-\delta k_{1},\]
which is a contradiction, since this objective function attains its maximum value at $(x_{1},k_{1})$. Hence $k_{1} > k_{1}^{'}$ is impossible. But since, $k_{1} \ge k_{1}^{'}$, hence we conclude that $k_{1}=k_{1}^{'}$
\end{proof}
\begin{lemma}
\label{One_Lemma_Two}
$E\left[X_{x_{1}}\right]-(1-k_{1}^{'})-\delta k_{1}^{'} \geq E\left[X_{x_{1L}}\right]-(1-k_{1L}^{'})-\delta k_{1L}^{'}$, where $k_{1L}^{'}=\max\left\{K(x_{1L}),k_{lev}\right\}$. 
\end{lemma}
\begin{proof}
The proof follows from the fact that LHS in the inequality, is the maximum value of the function $E\left[X_{x}\right]-(1-k)-\delta k$.
\end{proof}	
Now let us come to the main result.
\begin{theorem}
\label{One_Theo_Three}
In case of maximizing expected return,
Problem \ref{One_Prob_Three} outperforms Problem \ref{One_Prob_One}
\end{theorem}
\begin{proof}
Using Lemma \ref{One_Lemma_Two} we get, 
\[E\left[X_{x_{1}}\right]-(1-k_{1}^{'})-\delta k_{1}^{'} \geq E\left[X_{x_{1L}}\right] -(1-k_{1L}^{'})-\delta k_{1L}^{'}.\]
\begin{enumerate}
\item[Case 1:] We consider $k_{1}^{'}=k_{1L}^{'}=k_{lev}$, so that the above inequality becomes,
\[E\left[X_{x_{1}}\right] \geq E\left[X_{x_{1L}}\right].\]
Therefore by the Arbitrage Pricing theory (APT), $x_{1}$ contains more risky loans than $x_{1L}$. 
\item[Case-2:] Next we consider $k_{1}=k_{lev}$ and $k_{1L}=K(x_{1L})$. Now, since $k_{1L}=K(x_{1L})=\max\left\{K(x_{1L}),k_{lev}\right\} \geq k_{lev}$, therefore,
\[E\left[X_{x_{1}}\right]-(1-k_{lev})-\delta k_{lev}\geq E\left[X_{x_{1L}}\right]-(1-k_{lev})-\delta k_{lev},\]
as $(x_{1},k_{1})$ is the optimal solution. Hence, this implies that,
\[E\left[X_{x_{1}}\right] \geq E\left[X_{x_{1L}}\right].\]
Therefore by the APT, $x_{1}$ contains more risky loans than $x_{1L}$.
\item[Case-3:] We now consider $k_{1}=K(x_1)$ and $k_{1L}=k_{lev}$. 
Then
\[E\left[X_{x_{1}}\right]-(1-K(x_1))-\delta K(x_1) \geq  E\left[X_{x_{1L}}\right]-(1-k_{lev})-\delta k_{lev}.\] 
Now, 
\[E\left[X_{x_{1}}\right]-(1-k_{lev})-\delta k_{lev} \geq E\left[X_{x_{1}}\right]-(1-K(x_1))-\delta K(x_1),\]
since $k_{lev} \leq K(x_1)$ in this case. Now from the above two inequalities, we get,
\[E\left[X_{x_{1}}\right]-(1-k_{lev})-\delta k_{lev} \geq E\left[X_{x_{1L}}\right]-(1-k_{lev})-\delta k_{lev}.\] 
Therefore,
\[E\left[X_{x_{1}}\right] \geq E\left[X_{x_{1L}}\right].\]
Therefore by the APT, $x_{1}$ contains more risky loans than $x_{1L}$.
\item[Case-4:] We finally consider the case $k_{1}=K(x_1)$ and $k_{1L}^{'}=K(x_{1L})$. Accordingly, we get, 
\[E\left[X_{x_{1}}\right]-(1-K(x_1))-\delta K(x_1) \geq  E\left[X_{x_{1L}}\right]-(1-K(x_{1L})-\delta K(x_{1L}).\]
This implies that,
\[\sum\limits_{i}x_{1i}R_{i}-(1-K(x_1))-\delta K(x_1) \geq \sum\limits_{i} x_{1Li} R_{i}-(1-K(x_{1L}))-\delta K(x_{1L}).\]
Now we know that the function $\displaystyle{K(x)=\sum\limits_{i} x_{i} K_{i}}$, where $K_{i}$ is the capital requirement for the $i$-th loan \cite{Kiema2014}. Applying this in the inequality, we obtain,
\[\sum\limits_{i} (x_{1i}- x_{1Li})(R_{i}-(1-k_{i})-\delta k_{i}) \geq 0.\]
Now since $\displaystyle{\sum\limits_{i} x_{1i}=\sum\limits_{i}x_{1Li}=1}$, therefore some of $(x_{1i}- x_{1Li})$ are negative, while some are positive. As the overall sum is positive, hence the positive terms have more weightage. Hence $x_{1}$ has more risky investment than $x_{1L}$.
\end{enumerate}
Therefore we can see that portfolio $x_{1L}$ contains less riskier loans than portfolio $x_{1}$, thereby producing much profit upon success of the bank.
\end{proof}

\section{An Example}

Financial institutions classify the investment instruments (loans, in this case), contingent on its quality, particularly its creditworthiness. The ranking of the loan seekers is executed through various grades, with the prominent ones being the ratings of S\&P and Moody's. In the context of our discussion, we construct an example to illustrate our theoretical results, presented in the preceding Section. As a model built-up, we consider a basic scenario of two kinds of loans, namely, a safe loan and a risky loan. From the perspective of credit risk management, some of the factors which play a pivotal role in characterizing the loan are Probability of Default (PD), Exposure at Default (EAD), Loss Given Default (LGD), returns and the statutory capital requirements. The values of some of these parameters are available in articles \cite{Kiema2014,Shi2016} and also publicly available on the websites of banks (\cite{LBBW} and \cite{SBI}, for instance). Finally, motivated from \cite{Kiema2014}, we take the value of $\delta=1.04$.

For the illustrative example, we have constructed a loan portfolio of three loans, namely, one which is completely safe while the other two are risky, with one of these two being riskier than the other owing to greater the probability of default and expected loss. The value of all the parameters are enumerated in Table \ref{One_Table_Three}.
\begin{table}[h]
\centering 
\begin{tabular}{|c|c|c|c|}
\hline
Loan Type & Return & PD & LGD \\
\hline
Safe Loan & $r_{rf}=3\%$ & $0$ & $0$ \\
\hline
Less Risky Loan & $r_{s}=9\%$ & $p_{s}=6.1\%$ & $lgd_{s}=10\%$ \\
\hline
More Risky Loan & $r_{r}=13.2\%$ & $p_{r}=12.2\%$ & $lgd_r=9\%$ \\
\hline
\end{tabular}
\caption{Risk parameters for the  three loans}
\label{One_Table_Three}
\end{table}
Expected Loss plays the role of risk measure in our model. The formula \cite{chapter11} for the capital requirement for loans as a function of probability of default and loss given default, is given by,
\[C(PD,LGD)=LGD \times \left(Z-PD\right),\] 
where $Z$ is obtained as, 
\[Z=\left(\phi\left[\frac{\phi^{-1}(PD)+\sqrt{\rho}\phi^{-1}(0.999)}{\sqrt{1-\rho}}\right]\right).\]
Here $\phi$ is the cumulative standard normal distribution, while $\rho$ is different for different types of loans (the details of which are available in \cite{chapter11}). For this work, we have taken $\rho=0.15$. We solve the minimization of risk model by using ``scipy'' package of Python. Then, the model becomes,
\[\min_{x,k}~p_{s} \times lgd_{s} \times x_{1}+p_{r} \times lgd_{r}\times x_{2},\]
subject to:
\[x_{0}+x_{1}+x_{2}=1,~k\geq k_{lev},\]
\[k\geq \max\{C(p_{s}, lgd_{s}),k_{lev}\} \times x_{1}+\max\{C(p_{r}, lgd_{r}),k_{lev}\} \times x_{2}\]
and
\[\max(R_{4},0)(1-p_{s})(1-p_{r})+\max(R_{5},0)(1-p_{s})p_{r}+\max(R_{6},0)p_{s}(1-p_{r})+\max(R_{7},0)p_{s}p_{r}-\delta k \geq 0.098.\]
Here,
\begin{eqnarray*}
R_{4}(x_{0},x_{1},x_{2},k)&=&r_{rf}x_{0}+(1+r_{s})x_{1}+(1+r_{r})x_{2}-(1-k),\\
R_{5}(x_{0},x_{1},x_{2},k)&=&r_{rf}x_{0}+(1+r_{s})x_{1}+(1-lgd_{r})x_{2}-(1-k),\\
R_{6}(x_{0},x_{1},x_{2},k)&=&r_{rf}x_{0}+(1-lgd_{s})x_{1}+(1+r_{r})x_{2}-(1-k),\\
R_{7}(x_{0},x_{1},x_{2},k)&=&r_{rf}x_{0}+(1-lgd_{s})x_1+(1-lgd_{r})x_{2}-(1-k).
\end{eqnarray*}
Solving this, we get the loan portfolio allocation of $\left(5.72\%,13.37\%,80.91\%\right)$ with Leverage Ratio of $4\%$. On the other hand, solving for the model of minimizing risk, without limited liability, we get loan portfolio allocation of $\left(2.43\%,13.18\%,84.39\%\right)$, again with Leverage Ratio $4\%$. So for this particular example, there is a $3.68\%$ decrease in the expected loss due, resulting from the inclusion of limited liability.

For the next case, we illustrate the case of maximizing return with, as well as without limited liability. Accordingly, for this case, the model with limited liability will be transformed into a smooth optimization problem. We have given an upper bound of $1.2\%$ in Expected Loss for the problem. We have taken the Leverage Ratio to be $4\%$. Solving this, we get $12.35\%$ decrease in Expected Loss for changing this objective function in this example.

Now we come to the problem of maximizing return with and without limited liability. Then Problem \ref{One_Prob_One} is easily solved, since all the functions in its objective functional and the constraints, are smooth. However since the objective functional contains the ``$\max$'', therefore it is not differentiable. Accordingly, the steps of our methodology are as follows:
\begin{enumerate}[(1)]
\item Model Formulation: The model with limited liability (henceforth called as Model 1-L), which is not differentiable, is given by:
\[\max_{x,k}~E\left(\max(X_{x}-(1-k)),0\right)-\delta \times k,\]
subject to,
\[k \geq \max\{k_{lev},K(x)\},~\sum\limits_{i=0}^{2}x_{i}=1,~\text{Expected Loss} \leq 0.012~\text{and}~x_{i} \geq 0,~\forall~i = 0,1,2.\]
Therefore, we transform this model into another problem which is differentiable. To do this we include some new variables, namely $x_{4}$, $x_{5}$, $x_{6}$ and $x_{7}$, in order to handle the non-differentiability part. Accordingly, the new model (henceforth called as Model 1-L-NM) is given by,	
\[\max_{x,k}\left[x_{4}(1-p_{s})(1-p_{r})+x_{5}(1-p_{s})p_{r}+x_{6}p_{s}(1-p_{r})+x_{7}p_{s}p_{r}-1.04\times k\right],\]
subject to,
\[k \geq \max\{k_{lev},K(x)\},~\sum\limits_{i=0}^{2}x_{i}=1,~\text{Expected Loss} \leq 0.012~\text{and}~x_{i} \geq 0,~\forall~i = 0,1,2,\]
as well as,
\[x_{4}(x_{4}-R_{4})=0,~x_{5}(x_{5}-R_{5})=0,~x_{6}(x_{6}-R_{6})=0~\text{and}~x_{7}(x_{7}-R_{7})=0.\]
\item Transformation: The supremum of the objective functionals for the above problem has the same value. Let $obj_{1L}$ and $obj_{1LNM}$ be the objective functional of the above models. Therefore the objective functional of the first problem is given by,
\begin{eqnarray*}
obj_{1L}&=&\max\left(R_{4},0\right)\left(1-p_{s}\right)\left(1-p_{r}\right)+\max\left(R_{5},0\right)\left(1-p_{s}\right)p_{r}\\
&+&\max\left(R_{6},0\right)p_{s}\left(1-p_{r}\right)+\max\left(R_{7},0\right)p_{s}p_{r}-1.04 \times k.
\end{eqnarray*}
Let $X_{1L}=\left(x_{0},x_{1},x_{2},k\right)$ be the solution of Model 1-L. It also belongs to the feasible set of the Model 1-L-NM, with suitable values of $x_{4}$, $x_{5}$, $x_{6}$ and $x_{7}$.\\ Let $obj_{1LNM}$ attain its maximum at $X_{1LNM}=\left(x_{0}^{'}, x_{1}^{'}, x_{2}^{'}, k^{'}, x_{4}^{'}, x_{5}^{'}, x_{6}^{'}, x_{7}^{'}\right)$. Now, 
\begin{equation}
\label{One_Eq_Increasing}
\frac{\partial\left(obj_{1LNM}\right)}{\partial x_{i}} \geq 0,~ \forall~i=4,5,6,7,
\end{equation}
and the constraint involving $x_{i},~i=4,5,6,7$ does not involve $x_{j},~ j=4,5,6,7,~ \text{excluding}~i$. So any change in $x_{i}$ does not affect $x_{j}$. From equation \eqref{One_Eq_Increasing}, we see that $obj_{1LNM}$ increases with $x_{i},~i=4,5,6,7$. Hence $x_{i}$ takes the largest value in the region allowed by the constraint. Therefore, $x_{i}=\max(R_{i}(x_{0},x_{1},x_{2},k),0), ~i=4,5,6,7$. Consequently, we get ,
\[obj_{1LNM}(X_{1LNM})=obj_{1L}\left(x_{0}^{'}, x_{1}^{'}, x_{2}^{'}, k^{'}\right).\] 
As the restrictions on $\left(x_{0},x_{1},x_{2},k\right)$ are the same in both the problems, hence we obtain,
\begin{equation}
\label{One_L_Less_Than}
obj_{1L}(x_{0}^{'},x_{1}^{'},x_{2}^{'},k^{'}) \leq obj_{1L}(x_{0},x_{1},x_{2},k),
\end{equation}
since $\left(x_{0},x_{1},x_{2},k\right)$ maximizes the $obj_{1L}$.
Let us calculate $R_{i}(X_{1L}),~i=4,5,6,7$ and \\ $x_{i}=\max\left(R_{i}(x_{0},x_{1},x_{2},k),0\right),~i=4,5,6,7$. Further, we denote $X_{1L}^{'}=\left(x_{0},x_{1},x_{2},k,x_{4},x_{5},x_{6},x_{7}\right)$. As $X_{1LNM}$ maximizes the $obj_{1LNM}$, hence we get,
\begin{equation*} 
obj_{1LNM}(X_{1L}^{'}) \leq obj_{1LNM}(X_{1LNM}),
\end{equation*}
and,
\begin{equation}
\label{One_L_Greater_Than}
obj_{1L}(x_{0},x_{1},x_{2},k) \leq obj_{1L}(x_{0}^{'},x_{1}^{'},x_{2}^{'},k^{'}),
\end{equation} 
as, \[obj_{1LNM}\left(X_{1L}^{'}\right)=bj_{1L}\left(x_{0},x_{1},x_{2},k\right)~\text{and}~obj_{1LNM}(X_{1LNM})=obj_{1L}\left(x_{0}^{'},x_{1}^{'},x_{2}^{'},k^{'}\right).\] 
Therefore, from equations \eqref{One_L_Less_Than} and \eqref{One_L_Greater_Than}, we get that both the problems have the same supremum. If Model 1-L has a unique solution, then we get, $\left(x_{0},x_{1},x_{2},k\right)=\left(x_{0}^{'},x_{1}^{'},x_{2}^{'},k^{'}\right)$.
\end{enumerate}

In the example, we have taken three loan portfolios, one of which is completely risk-free. One of the remaining two is less risky (say $L_{s}$) and the remaining one is more risky (say $L_{r}$). In Figure \ref{fig:One_EL}, we have shown the change in Expected Loss against the change of investment in risky loans, where the $x$- axis presents the investment in $L_{s}$ and $y$ axis represents the investment in $L_{r}$. Also $\left(1-x-y\right)$ represents the investments in a completely safe loan. Therefore $\{(x+y \leq 1,~x \geq 0,~y \geq 0)\}$ contains all possible portfolios. 
From the image we see that increasing investments in risky loans increases the Expected Loss. Full investment in the completely safe loan gives an Expected Loss of $0\%$. In the case of full investment in $L_{s}$ and $L_{r}$, the resulting Expected Losses are $0.61\%$ and $1.10\%$, respectively.
%% For Expected Loss
\begin{figure}[h]
\centering
\includegraphics[width=0.6\linewidth]{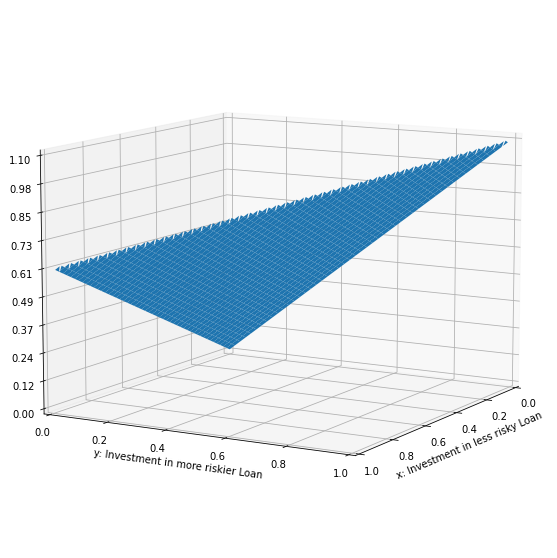}
\caption{Expected Loss against investments in risky loans}
\label{fig:One_EL}
\end{figure}
Now, we discuss the change in the expected loss. The sensitivity of unexpected loss is given in Figure \ref{fig:One_UL} with the investments in risky loans. Full investment in $L_{s}$ and $L_{r}$, results in Unexpected Loss of $2.39\%$ and $2.95\%$, respectively.
% For Unexpected Loss
\begin{figure}[h]
\centering
\includegraphics[width=0.6\linewidth]{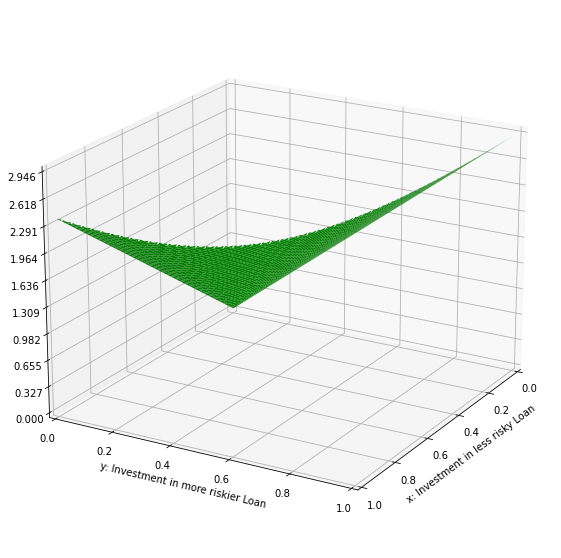}
\caption{Unexpected Loss against investments in risky loans}
\label{fig:One_UL}
\end{figure}
Next, we show the change of return with and without limited liability. We have shown the change of returns by fixing $x$ (investment in $L_{s}$) fixed at $0\%$ and $5\%$ in two images and then another two keeping $y$ (investment in $L_r$) fixed at $0\%$ and $10\%$. These four cases are done by keeping Leverage Ratio ($k$) at $4\%$ (Figure \ref{fig:One_LR4}), $7\%$ (Figure \ref{fig:One_LR7}) and $10\%$ (Figure \ref{fig:One_LR10}).
% For plotting K = 4%
\begin{figure}
\begin{subfigure}{.5\textwidth}
\centering
% include first image
\includegraphics[width=.8\linewidth]{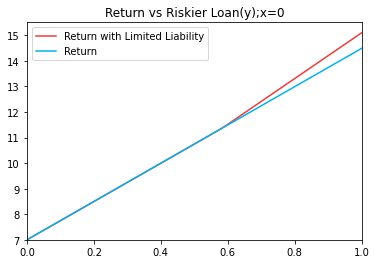}  
%\caption{Put your sub-caption here}
\label{fig:sub-first}
\end{subfigure}
\begin{subfigure}{.5\textwidth}
\centering
% include second image
\includegraphics[width=.8\linewidth]{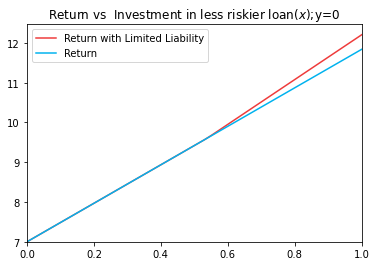}  
%\caption{Put your sub-caption here}
\label{fig:sub-second}
\end{subfigure}
\begin{subfigure}{.5\textwidth}
\centering
% include third image
\includegraphics[width=.8\linewidth]{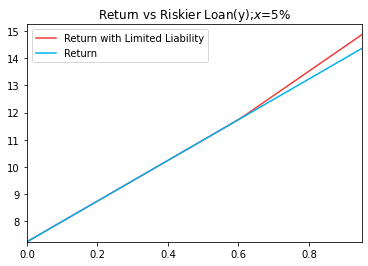}  
%\caption{Put your sub-caption here}
\label{fig:sub-third}
\end{subfigure}
\begin{subfigure}{.5\textwidth}
\centering
% include fourth image
\includegraphics[width=.8\linewidth]{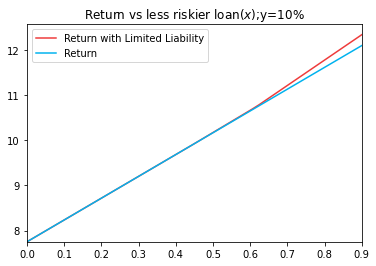}  
%\caption{Put your sub-caption here}
\label{fig:sub-fourth}
\end{subfigure}
\caption{Expected Returns versus risky loans keeping the Leverage Ratio at 4\%}
\label{fig:One_LR4}
\end{figure}

%% Plotting for K = 7%
\begin{figure}
\begin{subfigure}{.5\textwidth}
\centering
% include first image
\includegraphics[width=.8\linewidth]{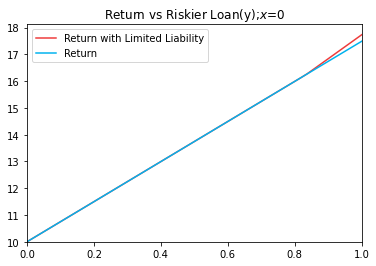}  
%\caption{Put your sub-caption here}
\label{fig:sub-first}
\end{subfigure}
\begin{subfigure}{.5\textwidth}
\centering
% include second image
\includegraphics[width=.8\linewidth]{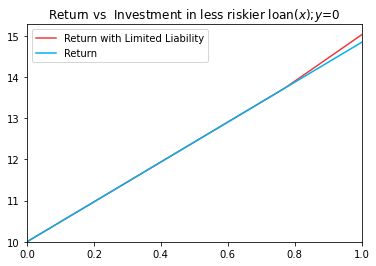}  
%\caption{Put your sub-caption here}
\label{fig:sub-second}
\end{subfigure}
\begin{subfigure}{.5\textwidth}
\centering
% include third image
\includegraphics[width=.8\linewidth]{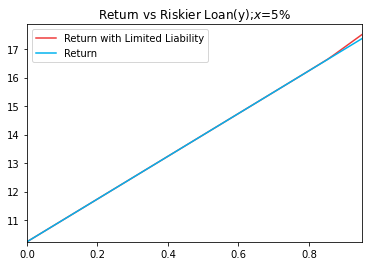}  
%\caption{Put your sub-caption here}
\label{fig:sub-third}
\end{subfigure}
\begin{subfigure}{.5\textwidth}
\centering
% include fourth image
\includegraphics[width=.8\linewidth]{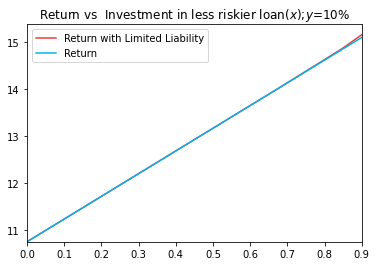}  
%\caption{Put your sub-caption here}
\label{fig:sub-fourth}
\end{subfigure}
\caption{Expected Returns versus risky loans keeping the Leverage Ratio at 7\%}
\label{fig:One_LR7}
\end{figure}

%% Plotting for K = 10% 
\begin{figure}
\begin{subfigure}{.5\textwidth}
\centering
% include first image
\includegraphics[width=.8\linewidth]{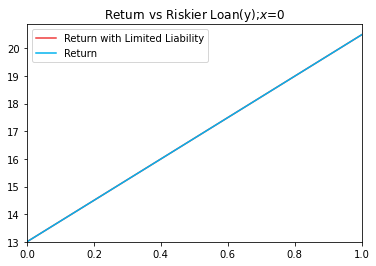}
%\caption{Put your sub-caption here}
\label{fig:sub-first}
\end{subfigure}
\begin{subfigure}{.5\textwidth}
\centering
\includegraphics[width=.8\linewidth]{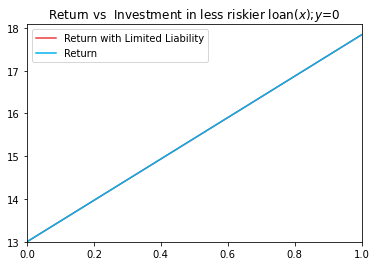}  
%\caption{Put your sub-caption here}
\label{fig:sub-second}
\end{subfigure}
\begin{subfigure}{.5\textwidth}
\centering
% include third image
\includegraphics[width=.8\linewidth]{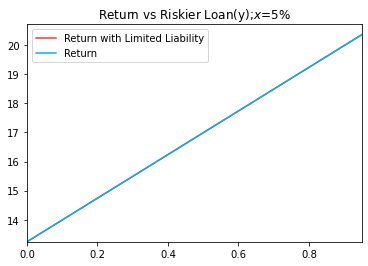}  
%\caption{Put your sub-caption here}
\label{fig:sub-third}
\end{subfigure}
\begin{subfigure}{.5\textwidth}
\centering
% include fourth image
\includegraphics[width=.8\linewidth]{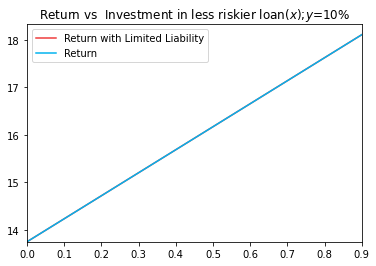}  
%\caption{Put your sub-caption here}
\label{fig:sub-fourth}
\end{subfigure}
\caption{Expected Returns versus risky loans keeping the Leverage Ratio at 10\%}
\label{fig:One_LR10}
\end{figure}

Now there are four scenarios of realizations (denoted by $s_{i},i=1:4$), for the two risky loans, namely, both the loans are repaid ($s_{1}$), only $L_{r}$ defaults ($s_{2}$), only $L_{s}$ defaults ($s_{3}$) and both the loans default ($s_{4}$). In Figure \ref{fig:One_LR4}, we have plotted the returns profile when the Leverage Ratio is $4\%$. In this case, let us assume that the bank has invested total wealth of $L_{r}$ ($x=0$). So, if $L_{r}$ defaults, the value of bank becomes $(1-lgd_{s})-(1-k)$, which is negative. Therefore the bank fails to meet its liabilities. In Figure \ref{fig:One_LR7}, we have plotted the returns profile when the Leverage Ratio is $7\%$. In this case, let us assume that the bank has invested total wealth to $L_{r}$. So, if $L_{r}$ defaults, the value of bank becomes $(1-lgd_{s})-(1-k)$, which is negative. Therefore, the bank again fails to meet the liabilities. The same consequences would be observed if the bank invests all its money to $L_{s}$ ($y=0$). When the Leverage Ratio is $7\%$, the loss in the worst case is less than that for the case in which the Leverage Ratio is $4\%$. If there is a penalty associated with the default and it is proportional to the amount of loss, then keeping a $7\%$ Leverage Ratio causes less penalty in case of failure. In Figure \ref{fig:One_LR10}, we have shown the change of return with keeping the Leverage Ratio at $10\%$. It is an interesting case, in the sense that even with the entire investments in $L_{s}$ or $L_{r}$, the bank can survive the worst cases, due to the Leverage Ratio being $10\%$. We can see that banks can survive all worst cases for the possible portfolios for this parameter value of the loans. From the graph, we see that there is no gap between the two lines. Therefore, an increment in the Leverage Ratio increases bank stability and consequently decreases in terms of the gap between the two lines.

\section{Conclusion}

Banks and other significant financial institutions have limited liability protection. However, in the literature discussed, the decision process usually goes through the classical framework, whereby the profit is measured via the Expected Return or utility functions of the realizations, in conjunction with the risk profiles. However, typically, the notion of limited liability is not considered in the model which is used to make the decision. To this end, this article is focused on the approach of incorporating limited liability in the model setup. We theoretically establish the benefit of limited liability for both the objectives of profit maximization and risk minimization. 

We have shown the comparative analysis between the models, for maximization of Expected Return, without limited liability (Problem \ref{One_Prob_One}) and with limited liability (Problem \ref{One_Prob_Three}), with a same upper bound on risk for both. Further, (Problem \ref{One_Prob_Two}) and (Problem \ref{One_Prob_Four}) minimizes risk without and with limited liability, respectively, while keeping identical lower bound for the Expected Return, for both.

Moreover, the analysis shows that in the case of the first comparison, the model with limited liability protection has less risky loans than the model without limited liability, which supports the result obtained in Theorem \ref{One_Theo_Three}. In the case of the second comparison, risk can be more minimized using the limited liability as the lower bound on the return with limited liability, as predicted in Theorem \ref{One_Theo_Two}. 

Therefore, from a practitioner's point of view, using limited liability in the decision model can help make the right investment decision, as limited liability protection reduces the amount of risk for the banks. So the construction of the portfolio, which meets the criteria of the investments ( threshold level of profit or risks ) and incorporates limited liability, is undoubtedly a more valuable and realistic model for the actual scenario.

\bibliographystyle{elsarticle-num}

\bibliography{BIBLIO}

\end{document}